\documentclass[letterpaper, 10 pt, conference]{ieeeconf}  

\IEEEoverridecommandlockouts                              

\usepackage{blindtext, graphicx}
\usepackage{amsmath}
\usepackage{url}
\usepackage{fancyhdr}
\usepackage{latexsym}
\usepackage[]{algorithm2e}
\usepackage{amssymb}
\usepackage{color}
\usepackage{comment}
\newtheorem{problem}{Problem}
\newtheorem{mydef}{Definition}
\newtheorem{theorem}{Theorem}

\newtheorem{example}{Example}
\usepackage{amssymb}
\usepackage{tikz}

\title{\LARGE \bf POMDP Model Learning for Human Robot Collaboration}

\author{Wei Zheng, Bo Wu and Hai Lin
\thanks{The partial support of the National Science Foundation (Grant No. CNS-1446288, ECCS-1253488, IIS-1724070) and of the Army Research Laboratory (Grant No. W911NF- 17-1-0072) is gratefully acknowledged.}
\thanks{The authors are with the Department of Electrical Engineering, University of Notre Dame, Notre Dame, IN, 46556 USA.
        {\tt\small wzheng1@nd.edu,bwu3@nd.edu,hlin1@nd.edu}}%
}

\begin{document}

\maketitle
\thispagestyle{empty}
\pagestyle{empty}

\begin{abstract}
Recent years have seen human robot collaboration (HRC) quickly emerged as a hot research area at the intersection of control, robotics, and psychology. While most of the existing work in HRC focused on either low-level human-aware motion planning or HRC interface design, we are particularly interested in a formal design of HRC with respect to high-level complex missions, where it is of critical importance to obtain an accurate and meanwhile tractable human model. Instead of assuming the human model is given, we ask whether it is reasonable to learn human models from observed perception data, such as the gesture, eye movements, head motions of the human in concern. As our initial step, we adopt a partially observable Markov decision process (POMDP) model in this work as mounting evidences have suggested Markovian properties of human behaviors from psychology studies. In addition, POMDP provides a general modeling framework for sequential decision making where states are hidden and actions have stochastic outcomes. Distinct from the majority of POMDP model learning literature, we do not assume that the state, the transition structure or the bound of the number of states in POMDP model is given. Instead, we use a Bayesian non-parametric learning approach to decide the potential human states from data. Then we adopt an approach inspired by probably approximately correct (PAC) learning to obtain not only an estimation of the transition probability but also a confidence interval associated to the estimation. Then, the performance of applying the control policy derived from the estimated model is guaranteed to be sufficiently close to the true model. Finally, data collected from a driver-assistance test-bed are used to train the model, which illustrates the effectiveness of the proposed learning method.

\end{abstract}

\section{INTRODUCTION}

Human-Robot Collaboration studies how to achieve effective collaborations between human and robots to synthetically combine the strengths of human beings and robots. While robots have advantages in handling repeated tasks with high precision and long endurance, human beings are much more flexible to changing factors which may introduce uncertainties that cost non-trivial efforts for robots to adapt. Therefore, to establish an efficient collaboration between human and robots is the core problem in the design of HRC system.

In recent years, Partially Observable Markov Decision Process (POMDP) models have emerged as one of the most popular models in HRC \cite{karami2010human}\cite{broz2011designing}\cite{nikolaidis2015efficient}. As a general probabilistic system model to capture uncertainties from sensing noises, actuation errors and human behaviors, POMDP model provides a comprehensive framework for the modeling and sequential decision making in HRC. The discrete states of the POMDP model can be used to represent robot and environment status and the hidden human intentions. In POMDP models, states are not directly observable but may be inferred by observations which represent the available sensing information regarding the statuses of human, robots, and the environment. Therefore, the partial observability on POMDP states models sensing noises and observation errors. Between different states, probabilistic transitions are triggered by various actions to describe uncertainties of the system actuation capabilities.

In this paper, we propose a framework to learn, through demonstrations, the HRC process which is modeled as a POMDP. A key challenge to learn a POMDP model from data is how to determine its hidden states.
Traditional approaches usually assume that the states in POMDP are given or the bound on the number of states is known. Instead, we drop these assumptions, as the states could be tedious to pre-define and the number of hidden states could be case dependent especially when human is involved. Hence, we propose to use a Bayesian non-parametric learning method to automatically identify the number of hidden states \cite{fox2009sharing}\cite{hughes2012effective}\cite{fox2014joint}.

Our proposed framework to learn a POMDP model is shown in Figure \ref{fig_diagram}. First, we assume the action set that represents the capability of the robot has been given. In each step, the robot chooses one action to collaborate with the human partner to accomplish some tasks. The training data can be collected by observing the physical system. After collecting enough training data, the structure (state space/observation space) of the POMDP model is determined from the data using a Bayesian non-parametric learning. In this paper, we propose to use the hidden Markov model (HMM) as the predictive model for the collected data. Each hidden state of the HMM represents a distinguishable motion pattern, and therefore corresponds to a state in POMDP. Taking advantages of the success of applying the Beta-Process Auto-regressive Hidden Markov Model (BP-AR-HMM) in human motion capture \cite{fox2014joint} and motion pattern recognition of humanoid robot for manipulation tasks \cite{wei2018solving}, we use this algorithm to automatically identify the number of hidden states of the HMM and learn the corresponding parameters.


Once the state/observation space is determined, the raw data are mapped to discrete observations using the maximum likelihood decision rule. Thus the discrete observation function can be obtained by calculating the decision error which only depends on parameters of the learned HMM model. Meanwhile, the sample mean is used to estimate the transition probability. Inspired by probably approximately correct (PAC) learning, we obtain not only an estimation of the transition probability but also a confidence interval associated to the estimation. Then the performance of applying the control policy derived from the estimated model is guaranteed to be sufficiently close to the optimal performance when the number of data for training is greater than a lower bound.

The main contribution of this paper is twofold. First, the Bayesian non-parametric learning method is used to learn the number of states automatically. Second, a lower bound on the number of training data is given. Once the number of training data is greater than this bound, the performance of the system from the optimal control policy based on the estimated model is guaranteed, with a high confidence, to be sufficiently close to the optimal performance of the true model.


\begin{figure}[!t]
\centering
\includegraphics[width=.8\linewidth]{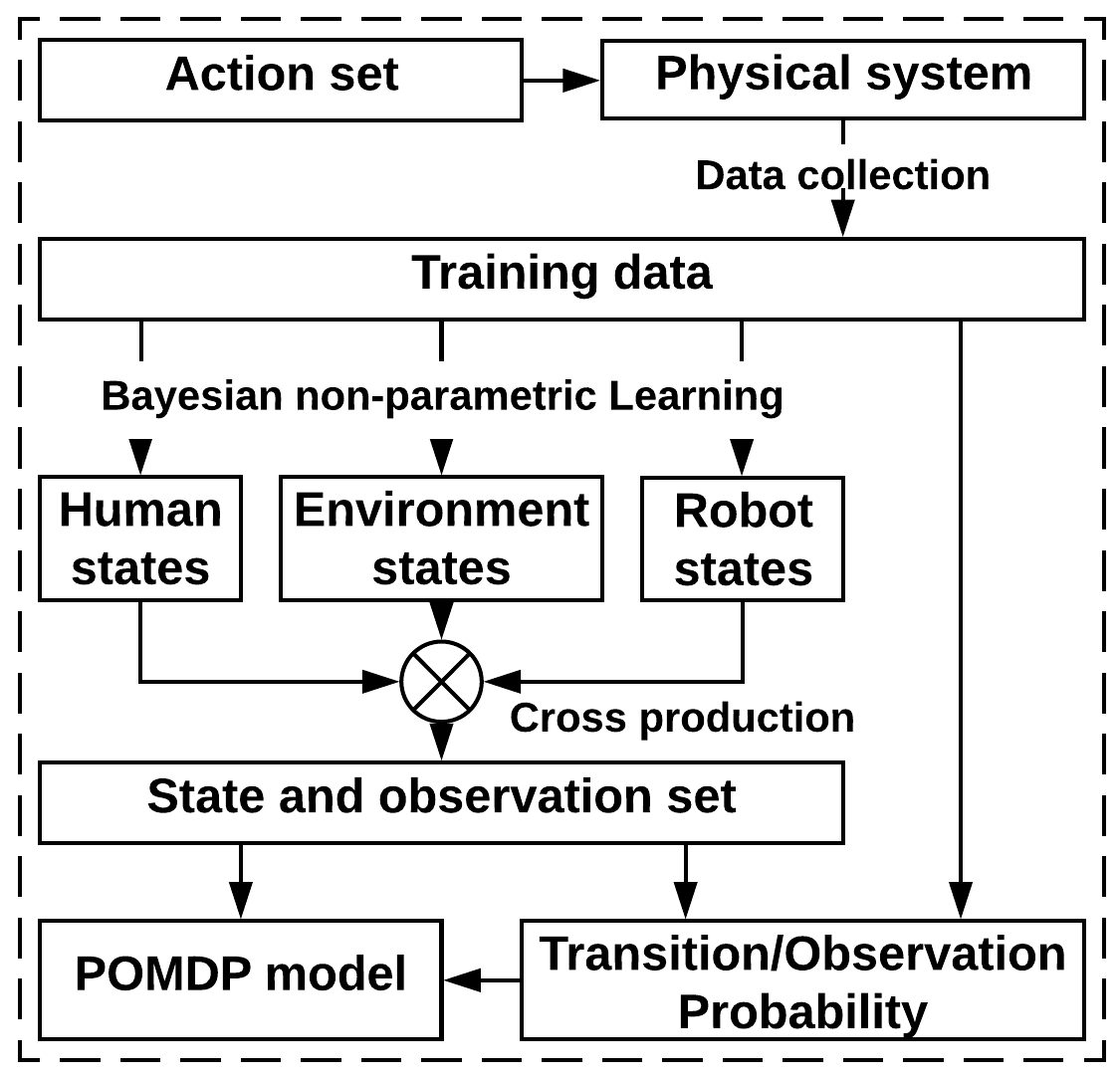}
\caption{Overview of the proposed framework. The training data are collected from the physical system. Using the training data, Bayesian non-parametric learning is used to learn the state/observation space of the POMDP automatically. The transition probability is estimated using the sample mean and the observation probability is obtained by calculating the decision error of the maximum likelihood decision rule.}
\label{fig_diagram}
\end{figure}






The rest of the paper is organized as follows. A summary of the related work is given in Section \ref{sec:relate}. Section \ref{sec:preliminaries} presents preliminaries on the POMDP model. Section \ref{sec:problem_formulation} formulates the problem. The learning framework with corresponding experiment results are shown in Section \ref{sec:mainresults}. Section \ref{sec:conc} concludes the paper.


\section{Related Work}
\label{sec:relate}
The POMDP model has received an increasing attention in the domain of HRC. In \cite{karami2010human}, the robot is capable of  predicting the goal of human by building a belief over his intentions. This prediction is integrated into a POMDP model where an appropriate decision can be solved to accomplish the shared mission. In \cite{broz2011designing}, human-robot social interactions are modeled as a POMDP where the intention of the human is represented as an unobservable part of the state space.  The ambiguous intentions of the human partner can be inferred through their observable behavior. The POMDP model is also used to predict the user's intention and get feedback from the user in terms of ``satisfaction" to enhance the quality of the interaction \cite{taha2011pomdp}. This strategy has been successfully applied in a robotic wheelchair and an intelligent walking device. In \cite{gopalan2015modeling}, POMDP is used to model joint tasks between a human and a robot. The POMDP state is a combination of the states of the world and the human mental state. The mental state is not visible to the robot, but it can be observed from the speech or gestures. In \cite{fern2014decision}, the hidden-goal MDP (HGMDP) model is introduced as a class of POMDP model, which is used to formalize and design an intelligent assistant. In \cite{wang2017anticipatory}, the problem of playing table tennis by human and robots is formulated as a POMDP model. The transition and observation are determined by the intention-driven dynamics model which allows the intention to be inferred from observed movements using Bayes' theorem \cite{wang2013probabilistic}.

Most of the aforementioned work assumes that the POMDP model is given. However, it is challenging to  get a POMDP model that realistically describes the proprieties of the system. A natural way is to learn the POMDP model from data. \cite{jaulmes2005active} studies active learning strategies for POMDPs. This method could handle significant uncertainties with relatively little training data. However, the need of an oracle
to answer the true identity of the hidden state could be too strong in some circumstances. From the statistical point of view, POMDP is an extension of HMM, which enable directly applying HMM learning algorithm to POMDP model learning when the action is fixed. For HMM model learning, the most standard algorithm for the unsupervised learning is the Baum-Welch algorithm \cite{bilmes1998gentle}. For the POMDP model learning, the Baum-Welch algorithm fixes the number of hidden states and begins with random initial conditions to learn the transition probabilities given the observation sequences. By using the Expectation-Maximization (EM) algorithm, a new state space can be updated to maximize the likelihood of the observations, then the Baum-Welch algorithm is executed again until the termination condition satisfied \cite{shani2007learning}. However, these learning procedures do not consider the complexity of the model and the model structure has to be specified in advance. This motivates us to apply the Bayesian non-parametric learning on POMDP model learning since it is then possible to extend hidden Markov models to have a countably infinite number of hidden states and the number of states can be inferred from data \cite{fox2009sharing}\cite{beal2002infinite}.



Theoretically, exact model learning through the observation and action sequences in HRC requires an infinite number of samples, which is not feasible in practice. Therefore, with only a finite number of samples, a salient approach is to estimate the actual POMDP within a certain precision and bound the optimality loss, which is the core idea of  PAC learning method. Existing results of PAC learning mostly focus on MDP models \cite{kearns2002near,fu2014probably,wu2017learning} and stochastic games \cite{brafman2002r,wen2016probably}. The main idea is to maintain and update an MDP model learned from observed state-action sequences through sampling. It can be proven that when the learning stops, with a predefined high probability, the estimated model approximates the actual model to a specified extent, such that the optimal policy found in the estimated model incurs a cost(reward) that is close to the optimal cost(reward) on the true model to a specified degree. Furthermore, the time, space and sampling complexity is polynomial with respect to the size of the MDP and measures of accuracy. However, for partially observable models, the existing results are only seen for HMMs \cite{gavalda2006pac}. In this paper, we are able to establish the PAC learning-like guarantee on POMDP models.


Distinct from most of the existing work on using POMDP models in HRC, we assume that the POMDP model including the structure of the model (states/observation space, transition/observation relations) are unknown, while most of the existing work assumes the number of states is given. Instead, we pursue a data-driven approach to learn such a model where the number of states is inferred from data without the need of prior knowledge and can be potentially infinite. Meanwhile, we provide a lower bound on the number of training data to guarantee the performance (optimality loss) satisfies the requirement.


\section{Preliminaries}\label{sec:preliminaries}

We begin with a brief description of the POMDP model for HRC.
\begin{mydef}\label{def:pomdp}
A POMDP model is defined as a tuple $P= (S, A, T, O , E, R)$
where $S$, $A$ and $O$ are sets of states, actions and observations. $T :S \times A \times S \to [0, 1]$ is a transition function and $E : S \times O \to [0, 1]$ is an observation function. The transition function $T(s'|s,a)$ defines the distribution over the next state $s'$ after taking an action $a$ from the state $s$. The observation function $E(o|s')$ is a distribution over the observation $o$ that may occur in the state $s'$. $R: S \times A \to \mathbb{R} $ is the reward function that represents the agent's preferences.
\end{mydef}

The control problem in POMDP targets on finding a policy that maximizes the expectation of cumulative rewards. Since states are not directly observable in POMDP model, the available information is observation-action sequence which is called history.

\begin{mydef}
Given a POMDP model $M$ and the corresponding history set $Hist$, the control policy $f : Hist \to A$ is a mapping from histories to actions.
\end{mydef}

Given a policy $f$, we can calculate the expected cumulative reward by the following equation.
\begin{equation}
    V= \sum_\rho p(\rho) R_c(\rho)
\end{equation}
where $\rho$ is a state-observation sequence, $p(\rho)$ is the probability of sequence $\rho$ and $R_c(\rho)$ is the cumulative reward when states and observations of the system follow sequence $\rho$.

The value $V$ depends on the transition probability, the observation probability and the reward function. Thus the accuracy of the model learned from data will directly influence the control policy solved and subsequently, influence the performance when applying the control policy on the system.


\begin{figure}[!t]
\centering
\includegraphics[width=.8\linewidth]{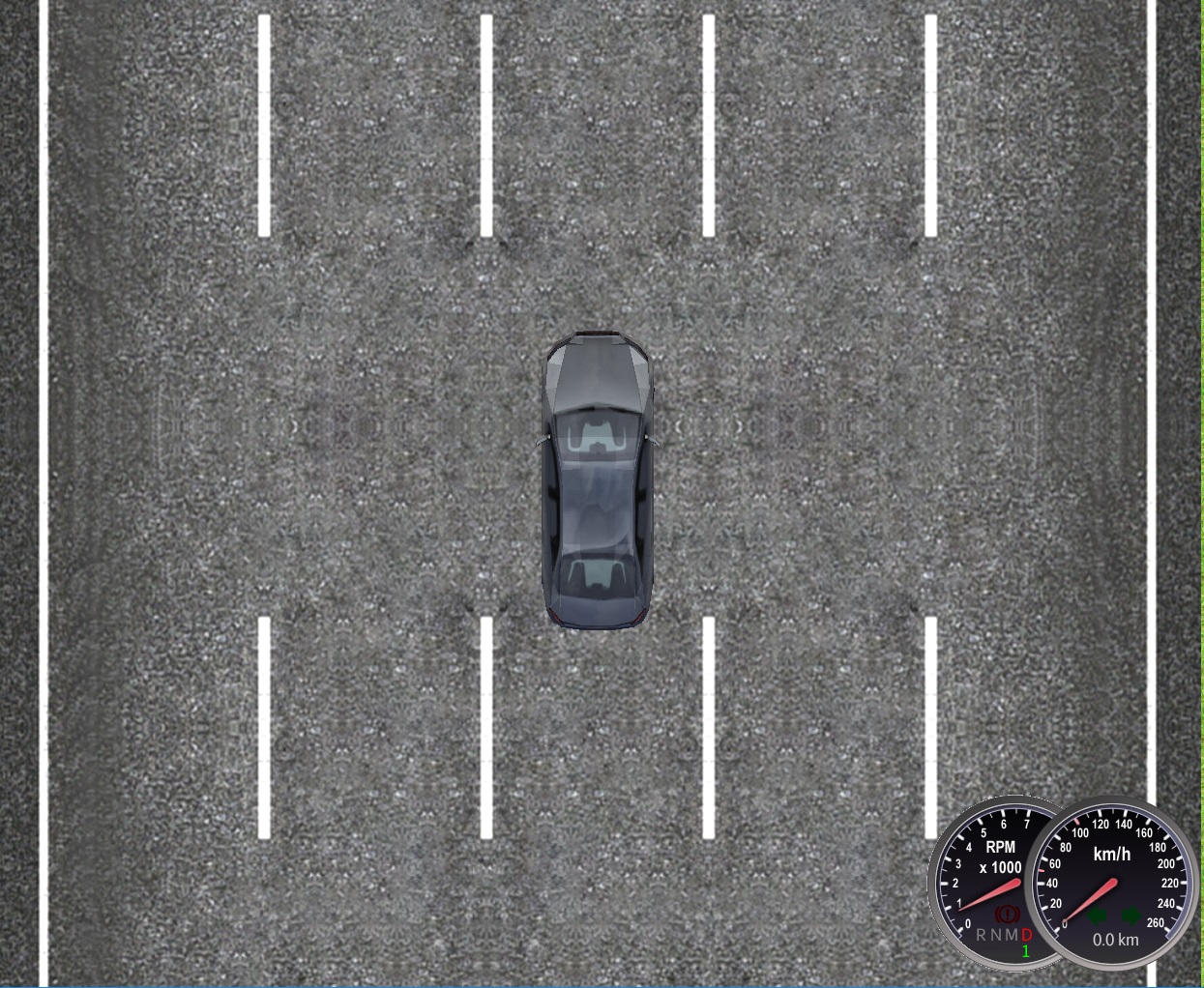}
\caption{The software OpenDs is used to simulate a vehicle driving on a five lane road. }
\label{fig_simulator}
\end{figure}

To quantify the closeness of POMDP models, we extend the definition of $\alpha$-approximation in MDPs to POMDPs \cite{kearns2002near}.

\begin{mydef}
($\alpha$-approximation of POMDP) Let $M$ and $\bar M$ be two POMDP models over the same state, action and observation space. $\bar M$ is an $\alpha$-approximation of $M$ if for any state $s$ and $s'$ and any action $a$, the inequality $\mathopen|T_M(s'|s,a)-T_{\bar M}(s'|s,a) \mathclose| \leq \alpha$ holds where $0<\alpha<1$. And for any state $s$, action $a$ and observation $o$, equalities $R_M(s,a)=R_{\bar M}(s,a)$ and $E_M(o|s)=E_{\bar M}(o|s)$ hold.
\end{mydef}

\section{Problem Formulation}\label{sec:problem_formulation}
In this paper, we consider the problem of learning the POMDP model for HRC and use the driver assistance system as a motivating example to illustrate the proposed approach.

\begin{example}
Consider the scenario where a human is driving on a five-lane road and the robot (vehicle) is to assist the driver to increase driving safety as shown in Figure \ref{fig_simulator}, we define the POMDP model for HRC in the following way.

The state space of the POMDP is the product of the state space of human, robot and environment, namely, $S=S_D \times S_R \times S_E$. The human state is his/her intention such as turning left/right or answering the cell phone. The state of the robot is the lane that the vehicle is currently driving on, namely $S_R=\{lane_1,...,lane_5\}$. The state of the environment indicates whether or not there are other vehicles around. Observation space of the POMDP is the same as the state space but inferred from the observed data. The action set of the robot can be advice, instructions and warnings to help the human driver to avoid potential dangers. The reward function is defined on the state and action pairs which quantifies the preference of actions in a specific state.
For example, if the current state is $s={(answerphone,lane_3,front)}$ which means that the human intends to answer the cellphone, the vehicle is in the lane $3$ and there is a vehicle in the front, the action that providing a warning is preferred. Thus the reward $R(s,``warning")$ is assigned to be a positive value.
\end{example}

Based on the POMDP model for HRC, an optimal control policy can be solved using POMDP planning algorithms by maximizing the expectation of cumulative rewards and subsequently, increase the driving safety. As the first step, to obtain a POMDP model is a critical problem for HRC.

\begin{problem}
Given the finite horizon $H$, upper bound of the reward function $R_{max}$, confidence level $\delta$ and a constant $\epsilon>0$, learn a POMDP model using the training data, such that the $H$ step expected cumulative reward is $\epsilon$ close to the optimal expected cumulative reward with confidence level no less than $\delta$.
\end{problem}

\section{Main Results}\label{sec:mainresults}
To obtain a POMDP model, all elements of the tuple in Definition \ref{def:pomdp} should be defined according to the application scenario. The action set represents the capability of the robot which has been determined when the robot is designed, and thus  is assumed to be known. The reward function is usually determined by the designer in the policy design process, so it is  assumed to be known as well. Therefore, in our paper, we mainly focus on determining and estimating the state/observation space, transition function and observation function.



\subsection{State space learning}\label{subsec:human_modeling}
Distinct from the majority of existing work on POMDP model learning from data, which usually assumes that the number of states is given, we propose to use Bayesian non-parametric learning to automatically identify the appearance of hidden states \cite{fox2014joint}. In this method, the BP-AR-HMM is used as the generative model of the data which is shown in Figure \ref{fig_bpar}. In the top layer, the Beta-Bernoulli process is used to generate an infinite number of features and to model the appearance of features among multiple time series. In the lower layer, the AR-HMM is used to model the relationship between the feature (hidden state) and the observed time series.


We begin with an introduction to the BP-AR-HMM model from the lower layer. For each time series $Y^i=[y^i_{1},...,y^i_{T_i}]$ where $y_t^i \in \mathbb{R}^{n}$ is a $n$ dimensional observed vector which could be eye movements, head motions and skeleton motions of the human, the generative model is described as
\begin{equation}
\begin{aligned}
   z_t^{i} & \sim \pi_{z_{t-1}^{i}}^{i}\\
   y_t^{i} & = \sum_{j=1}^{r} A_{j,z_t^{i}} y_{t-j}^{i}+e_t^{i}(z_t^{i})
    \end{aligned}
\end{equation}
where $e_t^{i}(k) \sim \mathcal{N}(0,\Sigma_k)$ is a Gaussian noise to capture the uncertainty and $r$ is the order of the AR-HMM. The variable $z^i_t$ is the hidden state and $\pi^i_j$ specifies the transition distribution of the state $j$. For each hidden state $z_t^i$, a set of parameters $\theta_{z_t^i}=\{A_{1,z_t^{i}},...,A_{r,z_t^{i}},\Sigma_{z_t^{i}}\}$ is used to characterize the corresponding motion pattern. Note that the HMM with Gaussian emissions is a special case of this model where the parameter will be adaptive to  $\theta_{z_t^i}=\{\mu_{z_t^{i}},\Sigma_{z_t^{i}}\}$. In this paper, the HMM model is used as the generative model since the behavior of human is constrained in a certain area in the driving scenario which makes the HMM sufficient to model the behavior.

This AR-HMM/HMM can be used to model only one time series while the training data will be multiple time series. Thus, the Beta-Bernoulli process is used to model the correlation between different time series in the top layer. This process is summarized as follows,
\begin{equation}
\begin{aligned}
   B|B_0 & \sim \text{BP}(c,B_0)\\
   X_i|B & \sim \text{BeP}(B)\\
   \pi_j^{i}| f_i,\gamma, \kappa & \sim \text{Dir}([\gamma,...,\gamma+\kappa,\gamma,...] \otimes f_i)
    \end{aligned}
\end{equation}
where $\text{BP}$ stands for Beta process, $\text{BeP}$ stands for Bernoulli process and $\text{Dir}$ stands for Dirichlet distribution.

A draw $B$ from a Beta process provides a set of global weights for the potentially infinite number of hidden states. For each time series $i$, an $X_i$ is drawn from a Bernoulli process parameterized by $B$. Each $X_i$ can be used to construct a binary vector $f_i$ indicating which of the global hidden state are selected in the $i^{th}$ time series. Then the transition probability vector $\pi^{i}_j$ of AR-HMM/HMM is drawn from a Dirichlet distribution with self-transition bias $\kappa$ for each state $j$.

Based on this generative model, parameters such as the hidden variable $z^i_t$ and $\theta_{z_t^i}$ can be inferred from data using the Markov chain Monte Carlo (MCMC) method.

\begin{figure}[!t]
\centering
\includegraphics[width=.8\linewidth]{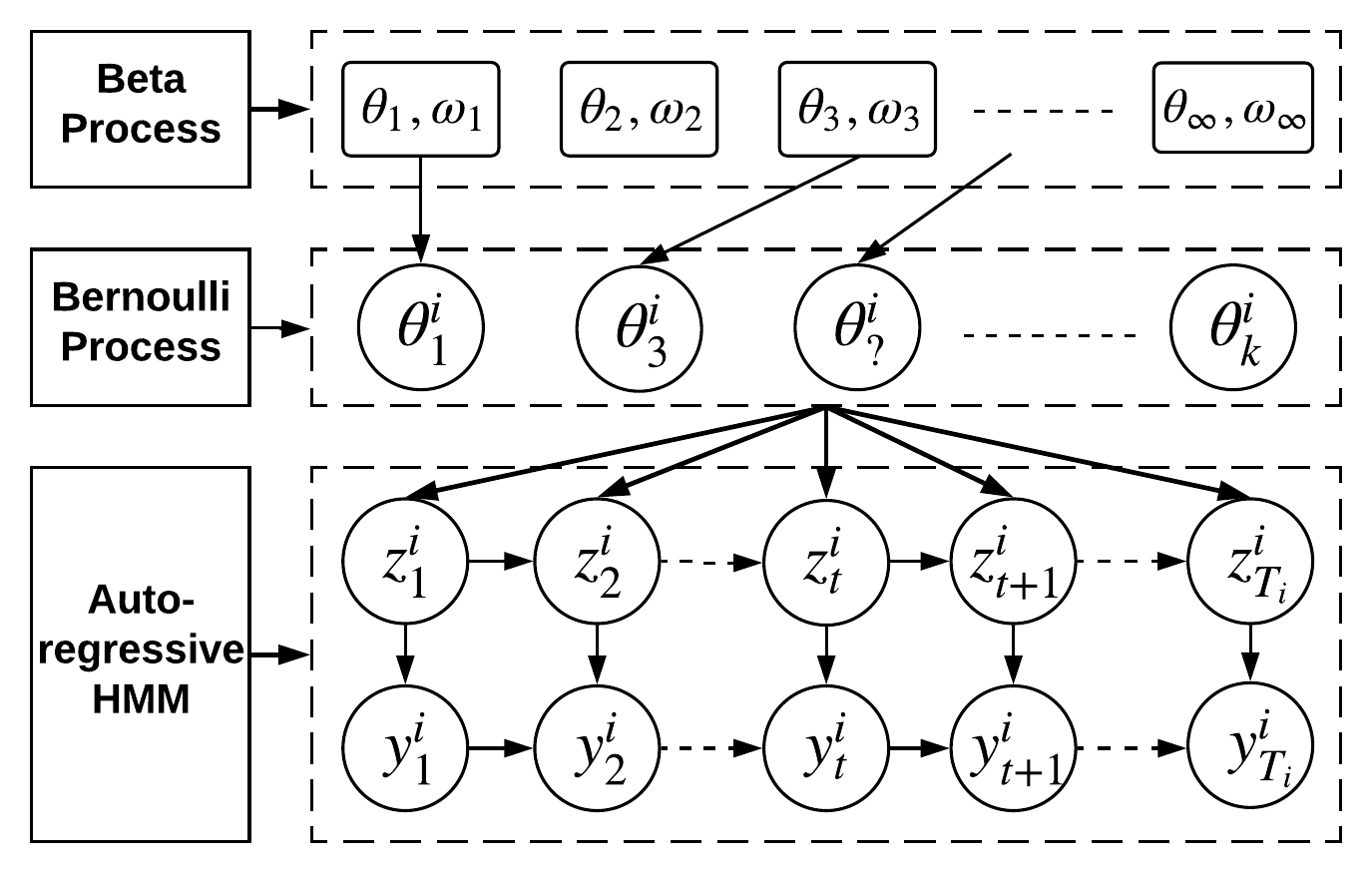}
\caption{The Beta process auto-regressive hidden Markov model. In the top layer, the Beta-Bernoulli process is used to generate an infinite number of features and to model the appearance of features among multiple time series. In the lower layer, the AR-HMM is used to model the relationship between the feature and the observed time series.}
\label{fig_bpar}
\end{figure}

Assume there are total $K$ time series in the training dataset $\{Y^1,...,Y^K\}$. After using the data set to train the BP-HMM model, for each data point $y_t^i$ in the trajectory $Y^i$, a hidden state $z_t^i$ is inferred to indicate under which hidden state the data is going to emit. The total number of hidden state in the model indicates the number of human motion patterns in the training dataset. Since the hidden state could exhibit among multiple time series, the same human motion pattern shared among these time series can be detected. Thus all possible human motion patterns that exhibit in the application scenario can be defined.

For convenience, we use $Z=\{z_1,...,z_{L}\}$ to represent the set of hidden states where $L$ is the total number of hidden states. After getting the set $Z$, the discrete state space of human $S_D$ can be defined as the same set of $Z$ with corresponding parameters $\theta_z$, namely $S_D=\{(z,\theta_z)|z \in Z\}$. If the state space of the robot $S_R$ and the environment $S_E$ has been defined using the same method, the state space of the POMDP model is simply the product of the human state space with the robot state space and the environment state space $S=S_D \times S_R \times S_E$. In this way, the number of states is inferred from the data without the need of prior knowledge and can be potentially infinite. It, therefore, provides a full Bayesian analysis of the complexity and structure of human models instead of defining them in advance. Meanwhile, the hidden parameter $\theta_z$ helps to characterize the mathematical property of the state and enable us to detect the hidden state from observations.

\begin{example}
As a running example, a driver and hardware-in-the-loop simulation system are used to validate the proposed approach. The software OpenDS\footnote{https://www.opends.eu/home} is used as the driving simulators and the Optitrack\footnote{http://optitrack.com/} system is used to capture the motion of the driver. By putting markers on the left/right hand of the driver and the steering wheel, a time series of positions are collected for each experiment. Each time series consists of driver behaviors such as turning left or right, answering phones and push a button on the laptop which is used to simulate operating instruments of the vehicle. The data from four experiments are used to train the BP-HMM model, which is shown in Figure \ref{fig:trajectory}. There are totally $1.4 \times 10^{4}$ data points in each time series, we only pick points from $1.1 \times 10^4$ to $1.4 \times 10^4$ to show.

\begin{figure}[t]
\includegraphics[width=1\linewidth]{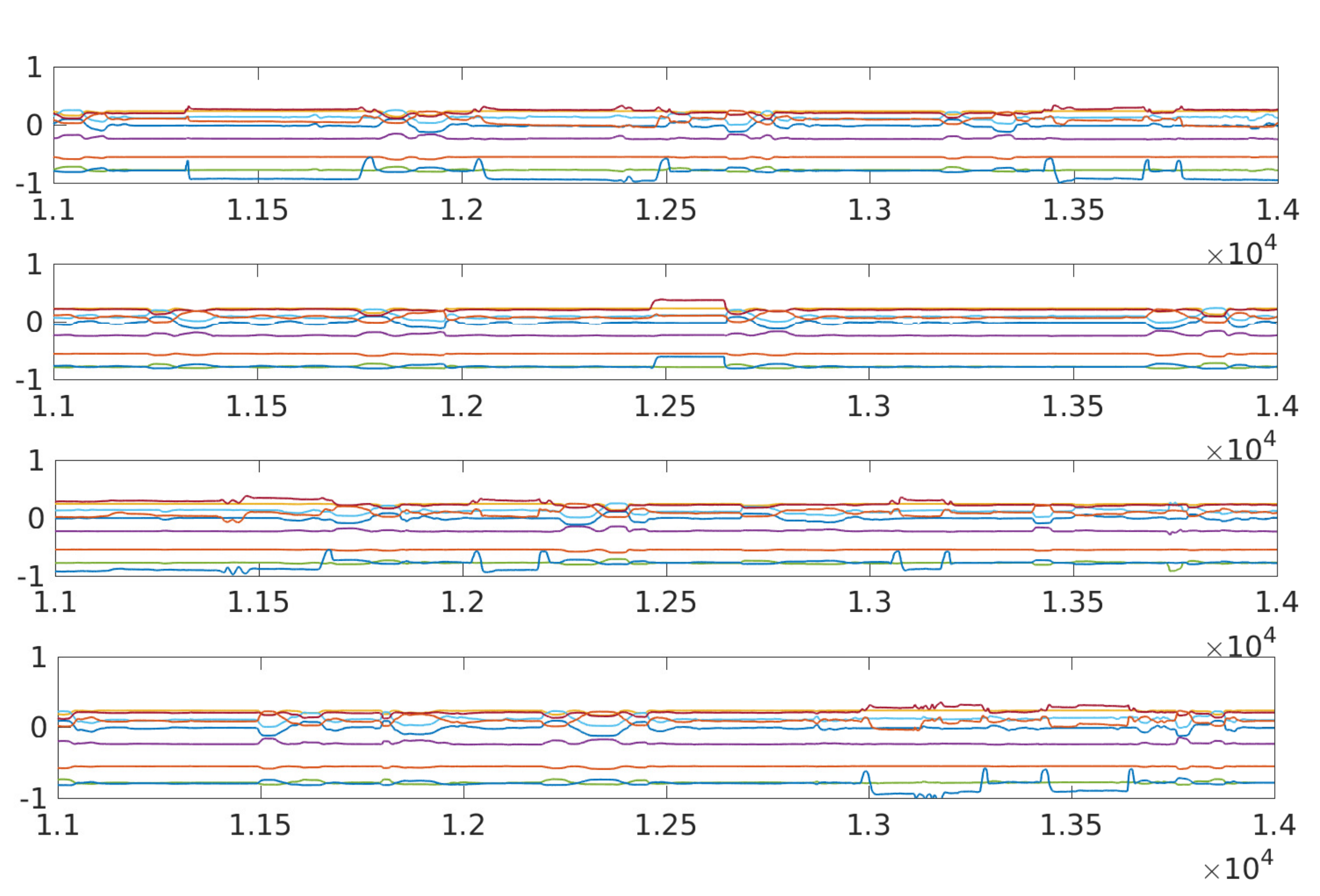}
\centering
\caption{Trajectories of 3D position of three markers. One marker is on the left hand of human and another is on the right hand. The third marker is on the steering wheel.}
\label{fig:trajectory}
\end{figure}

Then using the software provided by \cite{fox2014joint}, the learning result is shown in Figure \ref{fig:infer_intention}. In the figure, different hidden states are labeled by different colors. From the result, there are totally six motion features detected with the parameter $\theta_z$. Based on the learning result, each hidden state can be labeled with a physical meaning such as turning left/right.

Note that this state identification process needs to be done off-line by stored training data and a human supervisor is requested to label the identified state with physical meaning. The human supervisor only needs to label the hidden state, which is must easier than labeling each training data point. Actually, the need of labeling is only for further reward and control policy design, if the reward or preference can be learned from data, there is no need to label the identified state.

\begin{figure}[t]
\includegraphics[width=.9\linewidth]{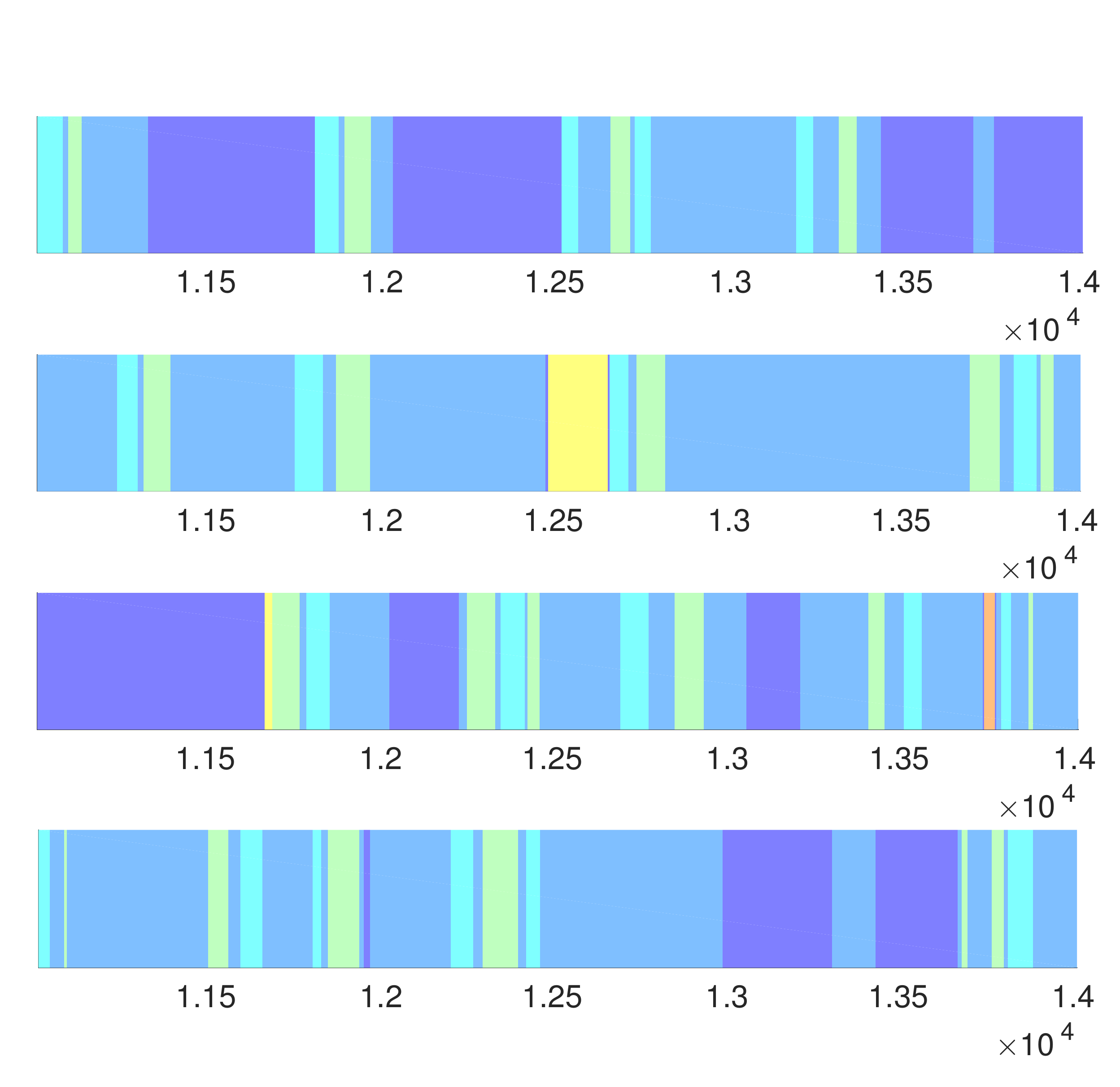}
\centering
\caption{The learning result of the Bayesian non-parametric method. There are six motion features detected which are distinguished by different colors. Correspondingly, the number of states for human in the POMDP model is six.}
\label{fig:infer_intention}
\end{figure}

\end{example}

\subsection{Observation function calculation}\label{subsec:obsservationfunction}
From the motion capture system, the observed data $y_t$ is a $n$ dimensional vector from continuous space. A decision rule is needed to map the observed data $y_t$ to discrete observation $o(t) \in O$ where $O$ is the observation set. In this paper, the state of robot $S_R$ and environment $S_E$ are assumed to be observed directly and correctly. This assumption could be dropped if the state of the robot and/or environment is learned in the way introduced in section \ref{subsec:human_modeling}. Thus the observation set $O$ is defined as the same set as $S_D$, namely $O=\{(z,\theta_z)|z \in Z\}$.

For each observation $y_t$, we propose to use the maximum likelihood (ML) estimator as the decision rule to map the observed data to one element of the observation set $O$. The decision rule is
\begin{equation}
    o(t)= \arg \max_{z_i \in Z}  L(y_t|s(t)=z_i)
\end{equation}
where $L(y_t|s(t)=z_i)$ is the likelihood of observing data $y_t$ when the state $s(t)$ is $z_i$. Since the emission of the HMM is a multivariate Gaussian distribution, the likelihood is
\begin{equation}
\begin{aligned}
    & L(y_t|s(t)=z_i)\\
    = & \frac{1}{\sqrt {(2\pi)^n \mathopen|\Sigma_{z_i}\mathclose|}} e^{-\frac{1}{2}(y_t-\mu_{z_i})^T \Sigma_{z_i}^{-1} (y_t-\mu_{z_i})}
\end{aligned}
\end{equation}

With the ML decision rule, the whole $n$ dimensional space is divided into $N$ disjoint regions $\Theta_j$ where
\begin{equation}
    \Theta_{j}=\{y_t| L(y_t|s(t)=z_j) \geq L(y_t|s(t)=z_{k}), \forall k \neq j\}.
\end{equation}
Then the decision rule is equivalent to
\begin{equation}
    o(t)= \{z_j|y_t \in \Theta_j\}.
\end{equation}

Due to the existence of the sensing noise, there is a certain detection error for the ML rule. Given the state $s(t)=z_i$, the probability of observing $o(t)=z_j$ is
\begin{equation}
    P(o(t)=z_j|s(t)=z_i)=\int_{\Theta_j} L(y_t|s(t)=z_i)
\end{equation}
Since the likelihood function is a multivariate Gaussian distribution function, the disjoint regions $\Theta_j$ is only determined by the parameters $\{(\mu_z,\Sigma_{z}), \forall z \in Z\}$ achieved in section \ref{subsec:human_modeling}. Thus it is reasonable to assert that the observation function of the POMDP model is fixed. However, to directly calculate the integration over region $\Theta_j$ is highly nontrivial since the function to be integrated is of high order and the integration boundary conditions can only be represented by inequality constraints. Therefore, in this paper, we follow the conventional way to use the Monte Carlo method to estimate the integration \cite{robert1998monte}.

Given the number of samples $N_{mc}$, we first sample from the multivariate Gaussian distribution and then decide which region the sample belongs to. Finally, the frequency of samples that exhibit in the region $\Theta_j$ is used to approximate the integration over $\Theta_j$.

\begin{example}
For example, if the current state is $s(t)=z_1$, the multivariate Gaussian emission parameter is ${(\mu_1,\Sigma_{1})}$. Using Monte Carlo method, $10^6$ vectors are sampled from the multivariate Gaussian distribution and then each vector is mapped to an observation using the ML rule. The frequency is used to approximate the observation probability. The observation probability for state $s(t)=z_1$ and $s(t)=z_4$ are shown in Figure \ref{fig:observation}.
\end{example}

\begin{figure}[t]
    \centering
    \includegraphics[width=1\linewidth]{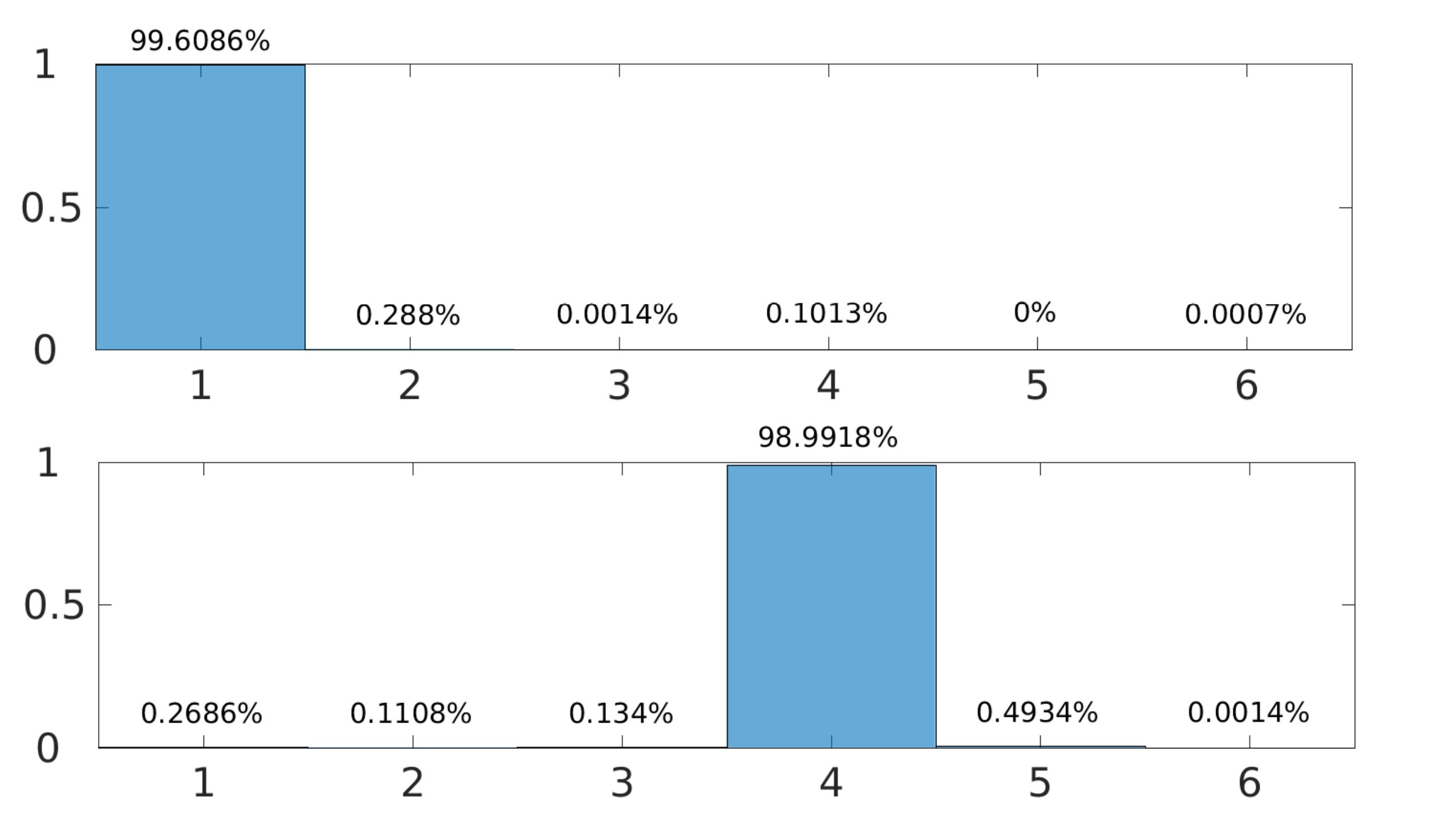}
    \caption{The observation probability calculated using Monte Carlo method when the state of POMDP is $s(t)=z_1$ and $s(t)=z_4$.}
    \label{fig:observation}
\end{figure}

\subsection{Transition function approximation}\label{subsec:transitionprobability}
In the driver assistance system, actions are designed to increase car safety and road safety by providing real-time advice, instructions and warnings or directly controlling the vehicles \cite{brookhuis2001behavioural}. If actions are designed in an advisory way, the transition probability captures uncertainties of the influence of actions on the human. Otherwise, if actions are designed in the automatic mode, the transition probability captures uncertainties from system actuation abilities. To learn the exact transition probability is difficult since the learned POMDP model may contain modeling uncertainties due to various reasons such as limited data and insufficient inference time \cite{lehner1996introduction}. In these cases, the modeling uncertainties will make the learned transition probabilities subject to a certain confidence level. This motivates us to apply the Chernoff bound to reason the accuracy of the transition probabilities for POMDP \cite{fu2014probably}.

Let $w(s'|s,a)$ denote the number of transitions observed in the data set from $s$ to $s'$ after taking action $a$ and $w(s,a) = \sum_{s'} w(s'|s,a)$ represents the total number of transitions start from $s$ for action $a$. Let $M$ denote the true model of the system which share the same structure of the learned model $\bar M$ . We use the sample mean to estimate the transition probability,
\begin{equation}
    T_{\bar M}(s'|s,a)= \frac{w(s'|s,a)}{w(s,a)}.
\end{equation}

The Chernoff bound gives the relation between the confidence level, the confidence interval and the number of data needed \cite{vadhan2012pseudorandomness}.
\begin{equation}
    P(|T_{\bar M} (s'|s,a)-T_{M} (s'|s,a)| \leq \alpha ) \geq 1-2 e^{-\frac{\alpha^2 w(s,a)}{2}}
\end{equation}
where $\alpha>0$. Given an accepted probability error bound $\alpha$ and a corresponding confidence level $\delta$, the number of data $w(s,a)$ should satisfy the following equation.
\begin{equation}\label{eq:num_data}
    w(s,a) \geq \lceil -\frac{2}{\alpha^2} \ln(\frac{1-\delta}{2}) \rceil.
\end{equation}

This bound gives a guidance for the transition probability training process. If the number of data collected for each state and action pair satisfies equation \ref{eq:num_data}, the POMDP model learned from data will be $\alpha$-approximation of the true model with confidence level no less than $\delta$.



\begin{example}
Assume the transition probability for state $s(t)=z_1$ is $[0.02,0.03,0.05,0.08,0.12,0.7]$. Let $\alpha=0.01$ and the confidence level $\delta=0.95$. The minimum number of data for training is $w(s,a)=73777$. We collect $73777$ samples from this distribution and use the sample mean to estimate the distribution. The sample mean is $[0.0203,0.0299,0.0510,0.0789,0.1177,0.7023]$. From the example, we see all the estimation fall into the confidence interval of the true distribution.
\end{example}

\subsection{Performance analysis using estimated model}\label{subsec:bound}
Since the true POMDP model is unknown, the control policy is designed only on the estimated model. A performance analysis is needed to evaluate the control policy.

In partially observable environments, the agent makes its decision based on the history of its actions and observations. A $H$-horizon policy tree is illustrated in Figure \ref{fig:policy}. Begin with a dummy node, the distribution of initial state is given as an initial belief $b_0$. In each step, an action is selected according to the observation of the system. After taking the action, the system makes a transition to another state in a stochastic way and then another action could be selected according to the new observation. This process will continue for $H$ steps.

Since the policy tree is calculated based on the estimated model $\bar M$ and applied to the real system or true model $M$, the performance of the system under the policy needs to be evaluated. In the following, we will discuss the finite step performance evaluation of the policy.

\begin{figure}[t]
\includegraphics[width=0.7\linewidth]{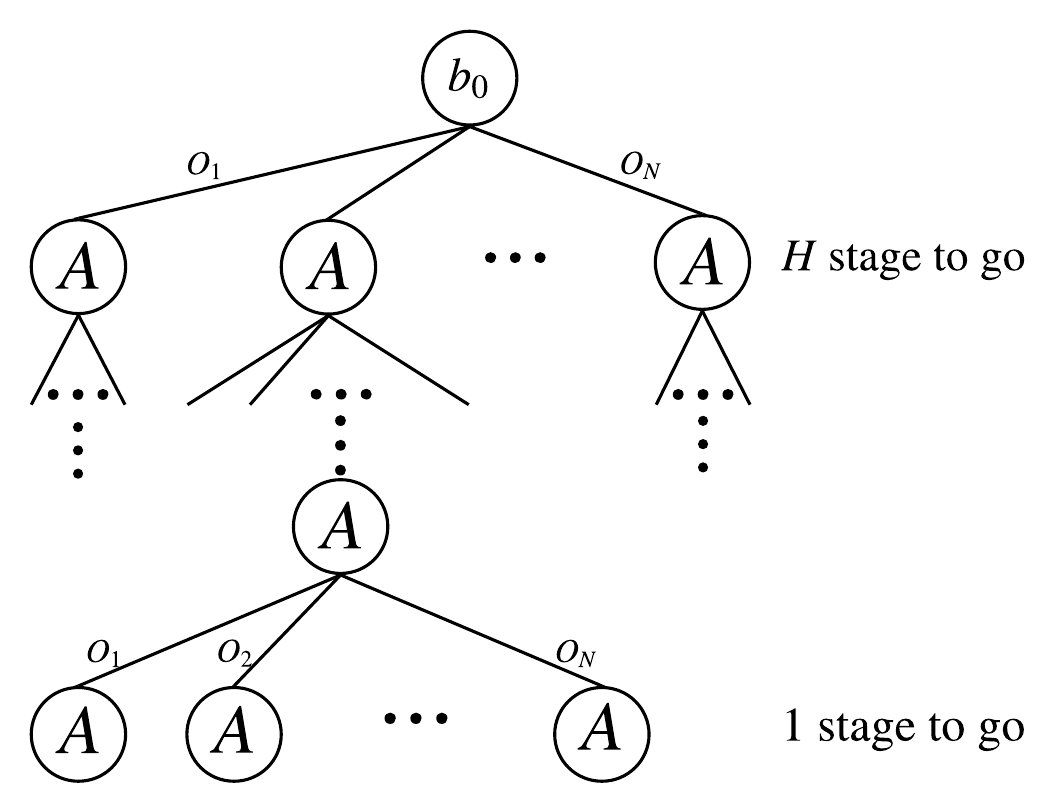}
\centering
\caption{A policy tree for horizon $H$. For each observation, there is a specific action selected from the action set $A$.}
\label{fig:policy}
\end{figure}

\begin{theorem}\label{thero1}
Given a POMDP model $\bar M$ which is an $\frac{\epsilon}{H^2 N R_{max}}$-approximation of the model $M$ where $H$ is a finite horizon, $N$ is the number of states, $R_{max}$ is the upper bound of the reward function and $0<\epsilon<1$. For any control policy $f$, the inequality $\mathopen|V^f_M-V^f_{\bar M}\mathclose| \leq \epsilon$ holds, where $V^f_M$ is the expected cumulative reward for the model $M$ under the control of the policy $f$.
\end{theorem}

\begin{proof}
For any control policy $f$ and any state-observation sequence $\rho=s_1,o_1,...,s_H,o_H$, we have
\begin{equation}
\label{error1}
\begin{aligned}
    &\mathopen|V^f_M-V^f_{\bar M}\mathclose|\\
    = & \mathopen|\sum_{\rho} p_{M}(\rho) R_c(\rho)- \sum_{\rho} p_{\bar M}(\rho) R_c(\rho)\mathclose|\\
    = & \sum_{\rho} \mathopen| p_{M}(\rho) - p_{\bar M}(\rho)\mathclose|\mathopen| R_c(\rho)\mathclose| \\
    \leq & \sum_{\rho} \mathopen| p_{M}(\rho) - p_{\bar M}(\rho)\mathclose| H R_{max}
\end{aligned}
\end{equation}
where $p_M(\rho)$ is the probability of sequence $\rho$ for the model $M$ and $R_c(\rho)$ is the cumulative reward for $\rho$ which is bounded by $H R_{max}$ where $R_{max}$ is the upper bound of the reward function.

The next step is to derive the bound of the probability error. Let $h_i$ denote the model where the transition probability is the same as $\bar M$ for the first $i$ step and the rest transition probability are identical to $M$. Thus $ p_{M}(\rho)=p_{h_0}(\rho)$ and $p_{\bar M}(\rho)=p_{h_H}(\rho)$. Then we have
\begin{equation}
\label{error2}
\begin{aligned}
    & \sum_{\rho} \mathopen| p_{h_0}(\rho) - p_{h_H}(\rho)\mathclose| \\
    = & \sum_{\rho} \mathopen| p_{h_0}(\rho) - p_{h_1}(\rho) + p_{h_1}(\rho) - \cdots + p_{h_H}(\rho)\mathclose| \\
    \leq & \sum_{i=0}^{H-1} \sum_{\rho} \mathopen| p_{h_i}(\rho) - p_{h_{i+1}}(\rho)\mathclose|
\end{aligned}
\end{equation}

Following the idea in \cite{fu2014probably}, we use $\tilde s$ to represent the state sequence and $\tilde o$ to represents the observation sequence.  Let $\hat s_i$ denote the $i$ step prefixes reaching $s_i$ and $\check s_i$ denote the suffixes starting at $s_i$ where $s_i$ is the state reached after $i$ step.
\begin{equation}
\label{error3}
\begin{aligned}
    & \sum_{\rho} \mathopen| p_{h_i}(\rho) - p_{h_{i+1}}(\rho)\mathclose|\\
    =  & \sum_{\tilde o} \sum_{\tilde s} \mathopen| p_{h_i}(\tilde o,\tilde s) - p_{h_{i+1}}(\tilde o,\tilde s)\mathclose|\\
    =  & \sum_{\tilde o} \sum_{\tilde s} \mathopen| p_{h_i}(\tilde o|\tilde s) p_{h_i}(\tilde s) - p_{h_{i+1}}(\tilde o|\tilde s)p_{h_{i+1}}(\tilde s)\mathclose|\\
    =  & \sum_{\tilde o} \sum_{s_i} \sum_{s_{i+1}} \sum_{\hat s_i} \sum_{\check s_{i+1}}  \mathopen|p_{h_i} (\hat s_i) p_{h_i}(\check s_{i+1}) p_{h_i}(s_{i+1}|s_i))\\
    & - p_{h_{i+1}} (\hat s_i) p_{h_{i+1}}(\check s_{i+1}) p_{h_{i+1}}(s_{i+1}|s_i))\mathclose| p_{h_i}(\tilde o|\tilde s)\\
    =  & \sum_{\tilde o} \sum_{s_i} \sum_{s_{i+1}} \sum_{\hat s_i} \sum_{\check s_{i+1}} p_{h_i} (\hat s_i) p_{h_i}(\check s_{i+1}) p_{h_i}(\tilde o|\tilde s) \\
    & \mathopen| p_{h_i}(s_{i+1}|s_i)-p_{h_{i+1}}(s_{i+1}|s_i) \mathclose|\\
    =  & \sum_{s_i} \sum_{s_{i+1}} \sum_{\hat s_i} \sum_{\check s_{i+1}} p_{h_i} (\hat s_i) p_{h_i}(\check s_{i+1}) \sum_{\tilde o} p_{h_i}(\tilde o|\tilde s) \frac{\epsilon}{H^2 N R_{max}} \\
    =  & (\sum_{s_i} \sum_{\hat s_i} p_{h_i} (\hat s_i) )( \sum_{s_{i+1}}  \sum_{\check s_{i+1}} p_{h_i}(\check s_{i+1})) \frac{\epsilon}{H^2 N R_{max}}\\
    = & \frac{\epsilon}{H^2 R_{max}}
\end{aligned}
\end{equation}
Thus combine results of equation \ref{error1}, \ref{error2} and \ref{error3}, we have
\begin{equation}
    \mathopen|V^f_M-V^f_{\bar M}\mathclose| \leq H^2 R_{max} \frac{\epsilon}{H^2 R_{max}} = \epsilon.
\end{equation}
\end{proof}

Theorem \ref{thero1} allows us to evaluate the performance distance with respect to the similarity of models. It guarantees that the expected cumulative reward of $\bar M$ is $\epsilon$ close to that of $M$ once the model $\bar M$ is $\frac{\epsilon}{H^2 N R_{max}}$-approximation of the model $M$.

However, this theorem is derived using the same control policy for different models. In reality, the control policy is solved using the estimated model and applied on the true model. The optimality loss should be the difference of the achieved reward when applying the derived control policy and the optimal policy to the true model.

\begin{theorem}\label{theor2}
Given a POMDP model $\bar M$ which is an $\frac{\epsilon}{2 H^2 N R_{max}}$-approximation of the model $M$ where $H$ is a finite horizon, $N$ is the number of states, $R_{max}$ is the upper bound of reward function and $0<\epsilon<1$. Let $f$ and $g$ be the $H$-step optimal control policy for model $\bar M$ and model $M$ respectively. Then the inequality $\mathopen|V^g_M-V^f_{M}\mathclose| \leq \epsilon$ holds where $V_M^f$ is the expected reward of the model $M$ under the control of the policy $f$ .
\end{theorem}

\begin{proof}
From theorem \ref{thero1}, we have $\mathopen|V^f_M-V^f_{\bar M}\mathclose| \leq \frac{\epsilon}{2}$ and $\mathopen|V^g_M-V^g_{\bar M}\mathclose| \leq \frac{\epsilon}{2}$. Since $f$ is the optimal policy for model $\bar M$, we have $V^g_{\bar M} \leq V^f_{\bar M}$. Then
\begin{equation}\label{inequality1}
    V^g_{ M} \leq V^g_{\bar M}+\frac{\epsilon}{2} \leq V^f_{\bar M}+ \frac{\epsilon}{2} \leq V^f_{M}+ \epsilon
\end{equation}
Since $g$ is the optimal control policy of $M$, we have $V^f_{M} \leq V^g_{M}$. Then
\begin{equation}\label{inequality2}
   V^f_{M} \leq V^g_{M} \leq V^f_{M}+ \epsilon
\end{equation}
Which implies the inequality $\mathopen|V^g_M-V^f_{M}\mathclose| \leq \epsilon$ holds.
\end{proof}

Theorem \ref{theor2} gives a relationship between the optimality loss and the similarity of models. Together with equation \ref{eq:num_data}, it gives a lower bound on the number of data needed to train the POMDP model. Conversely, if the number of data used to train the model is greater than the bound, the $H$-step expected cumulative reward is $\epsilon$ close to the optimal expected cumulative reward with confidence level no less than $\delta$.


\section{CONCLUSIONS}
\label{sec:conc}
In this paper, we proposed a framework to learn the POMDP model for HRC through demonstrations. Distinct from most of the existing work, we did not assume any knowledge on the  structure of the model. Instead, we proposed to use Bayesian non-parametric methods to automatically infer the number of states from the training data. When learning the transition probability, we provided a lower bound on the number of training data which guarantees the optimality loss is bounded using a control policy derived from the estimated model. In this paper, the observation probability was assumed to be precisely calculated using Monte Carlo approach. However, the observation probability was also subject to a confidence interval. Hence, taking the uncertainty of observation function into consideration will be our immediate future work.

\addtolength{\textheight}{-12cm}   






\bibliographystyle{unsrt}
\bibliography{ref}
\end{document}